\documentclass[11pt,reqno]{amsart}
\usepackage[foot]{amsaddr}
\usepackage{graphicx}
\usepackage{amsthm,amssymb,amsmath,calc,eucal,ifthen}
\usepackage{amscd}
\usepackage[all,arc,curve,color,frame,matrix,arrow]{xy}
\usepackage{xcolor}
\usepackage{pgf,tikz,pgfplots}
\pgfplotsset{compat=1.15}
\usepackage{mathrsfs}
\usetikzlibrary{arrows}
\numberwithin{equation}{section}

\newtheorem{theorem}{Theorem}[section]
\newtheorem{lemma}[theorem]{Lemma}
\newtheorem{proposition}[theorem]{Proposition}

\theoremstyle{remark}
\newtheorem{remark}[theorem]{Remark}
\newtheorem{example}[theorem]{Example}

\newtheoremstyle{rmdefinition}{}{}{\upshape}{}{\bfseries}{.}{ }{}

\theoremstyle{rmdefinition}
\newtheorem{definition}[theorem]{Definition}

\newcommand{\R}{{\mathbb R}}

\newcommand{\be}[1]{\begin{equation}\label{#1}}
\newcommand{\ee}{\end{equation}}
\newcommand{\beqa}{\begin{eqnarray}}
\newcommand{\eeqa}{\end{eqnarray}}

\newcounter{tmpc}
\newlength{\tmplenght}
\newlength{\tmplenghta}
\newlength{\tmplenghtb}
\newlength{\tmplenghtc}

\DeclareMathOperator{\Id}{Id}

\begin{document}

\title[Projective limits]{Projective limits in Euclidean quantum field theory, II: Abelian gauge theory}

\author{Svetoslav Zahariev}
\address{MEC Department, LaGuardia Community College of The City University of New York, 31-10 Thomson Ave., Long Island City, NY 11101, U.S.A.}
\email{szahariev@lagcc.cuny.edu}

\maketitle
\begin{abstract}
We present two constructions of continuum and thermodynamic limits of Abelian polyhedral gauge theories in arbitrary spacetime dimension. The first construction relies on the existence of projective systems of heat kernel measures, while the second involves an infinite dimensional heat kernel measure defined using projective limit of Hilbert spaces. As a special case, we obtain a model of Abelian gauge theory on the infinite cubical lattice, which, in contrast to the standard one, is massless for arbitrary values of the coupling parameter.
\end{abstract}
\section{Introduction}

In this article, we continue our study of continuum and infinite volume limits in Euclidean quantum field theory via projective systems of configuration spaces and probability measures initiated in \cite{Z1} in the context of Abelian gauge theory. 

The starting point of our constructions is the observation that the well-known Villain action in lattice (and more generally, polyhedral) gauge theory may be interpreted as the composition of the coboundary operator acting on group-valued 1-cochains with the heat kernel on the Lie group of 2-cochains. In the case of Abelian gauge theory, the image of the latter coboundary operator is a closed Lie subgroup of the group of all 2-cochains and we propose to modify the Villain action by composing the coboundary with the heat kernel on that image subgroup. This easily leads, in the case of a contractible polyhedral complex, to an interpretation of the measure associated to this modified action as a pushforward of the heat kernel measure on the image subgroup via the inverse of the map on gauge equivalence classes on 1-cochains induced by the coboundary operator.

The pushforward measure interpretation leads, in turn, to our first main result, Theorem \ref{mainconstrthe}, which may be stated as follows. Consider a countable sequence of finite contractible polyhedral complexes of dimension $d$ associated to either a continuum or an infinite limit procedure. Then there is an inductive construction of ``renormalized'' inner products on the Lie algebra valued 2-cochains on those complexes such that the associated modified Villain action measures form a projective system, hence a projective limit measure exists on the corresponding projective limit space of gauge equivalence classes of 1-cochains. When $d=2$ the natural Euclidean inner products suffice and our result reduces to that obtained in \cite{VM0}-\cite{VM2}. It should also be pointed out that results on the convergence of Abelian lattice gauge theory to its continuum limit within the standard constructive QFT framework have been obtained for $d=3$ and $d=4$, cf. \cite{Gro} and \cite{Dr}.

When the above construction is applied to infinite volume limits in dimensions larger than 2, one obtains a measure on the infinite lattice $\mathbb{Z}^d$ which unfortunately is not invariant under lattice translations. To remedy this deficiency, we propose an alternative construction which yields a modified Villain action measure directly on the projective limit configuration space. This is achieved by obtaining first an infinite dimensional heat kernel measure on the projective limit space of exact 1-cochains and then taking its pushforward under the inverse of the limit coboudary operator. According to our second main result, Theorem \ref{masslessth}, the infinite lattice measure so obtained is lattice translation invariant and moreover corresponds to a massless field for arbitrary values of the coupling parameter and for any $d>2$. We recall that, in contrast, the standard thermodynamic limit of 4-dimensional Abelian lattice gauge theory exhibits a phase transition and is massless only for small couplings, as proved in \cite{OS} and \cite{FS}.

This article is organized as follows. In Section \ref{prelimsec} we present background material mostly concerning Haar measures, heat kernels and Fourier analysis on compact Abelian Lie groups. In Section \ref{vactinpoly} we introduce the modified Villain action described above and discuss its basic properties. Section \ref{projmlimmeasec} is dedicated to the construction of projective systems of measures via ``renormalized'' inner products. Finally, in Section \ref{limmeassect} we introduce a general construction of limit heat kernel measures on projective limits of compact Abelian Lie groups and then establish the masslessness property of the associated infinite volume Abelian lattice gauge theory.

\section{Preliminaries}\label{prelimsec}
\subsection{Fourier transform on compact Abelian groups} In this section we briefly review some basic facts from harmonic analysis on Abelian groups. For more details the reader is referred to \cite{Ru} or \cite{Fo}.
Let \( G \) be an Abelian compact topological group. Recall that its dual 
\emph{dual group} \( \widehat{G} \) may be defined as the set of all continuous group homomorphisms (characters) $\chi$ from $G$ to the circle group $U(1)$,
equipped with pointwise multiplication:
\[
(\chi_1 \chi_2)(x) = \chi_1(x) \chi_2(x).
\]
If \( F: G_1 \to G_2 \) is a continuous homomorphism of compact Abelian groups, the \emph{dual map}
$
\widehat{F} : \widehat{G}_2 \to \widehat{G}_1
$
is given by
\[
\widehat{F}(\chi) = \chi \circ F, \quad \chi \in \widehat{G}_2.
\]
We write $\nu$ for the normalized Haar measure on $G$ and, given $f \in  L^1(G,\nu)$ define its \emph{Fourier transform} via
\[
\widehat{f} : \widehat{G} \to \mathbb{C}, \quad \widehat{f}(\chi) = \int_G f(x) \overline{\chi(x)} \, d\nu(x).
\]
More generally, the Fourier transform of a finite Borel measure \( \mu \) on \( G \) is defined by
\[
\widehat{\mu}(\chi) = \int_G \overline{\chi(x)} \, d\mu(x), \quad \chi \in \widehat{G}.
\]
Now let \( G_1, G_2 \) be compact Abelian groups and let \( \mu_i\) be a finite Borel measure on $G_i$. Let \( F: G_1 \to G_2 \) be a continuous homomorphism and denote the pushforward/image measure of \( \mu_1\) via $F$ by $F(\mu_1)$.
\begin{lemma}\label{pushfofour}
One has \( F(\mu_1) = \mu_2 \) if and only if 
\(\widehat{\mu_2} \,=\, \widehat{\mu_1} \circ \widehat{F}\).
\end{lemma}
\begin{proof}
Applying the change of variables formula for pushforward measures, we find
$$
\int_{G_2} \overline{\chi(y)}\,dF(\mu_1)(y)
\;=\;
\int_{G_1} \overline{\chi(F(x))}\,d\mu_1(x)
\;=\;
\widehat{\mu_1}(\chi\circ F)
$$
for all $\chi \in \widehat{G}_2$. Since the leftmost integral above is equal to \(\widehat{F(\mu_1)}(\chi)\), we conclude that \(F(\mu_1)=\mu_2\) implies 
\(\widehat{\mu_2} \,=\, \widehat{\mu_1} \circ \widehat{F}\).
The converse implication follows from the fact that a finite Borel measure is determined by its Fourier transform (cf. \cite[(4.33)]{Fo}).
\end{proof}

\subsection{Heat kernels and heat kernel measures}\label{hkermeassec}
Recall that the heat kernel $p_{\beta}(x,y)$ on a compact connected Riemannian manifold $M$ is defined as the kernel of the operator $e^{-\beta \Delta_M}$, where $\beta>0$ and $\Delta_M$ stands for the (positive) Laplacian on $M$. The function $p_{\beta}(x,y)$ is smooth, symmetric and strictly positive for all $\beta>0$. For the construction of the heat kernel and its properties, the reader is referred to \cite[Chapters 10 and 11]{Li} and \cite[Chapter 3]{Ro}.

Let $G$ be a compact connected Lie group equipped with a left-invariant Riemannian metric. We shall call the {\em heat kernel} on $G$ the function $H_{\beta}$ on $G$ given by evaluating one of the arguments of $p_{\beta}(x,y)$ at the identity of $G$. More information about heat kernels on Lie groups may be found in \cite{VSC} and \cite{Ma}.  Further, we shall call the Borel probability measure
$$ \mu_{\beta}=z_{\beta} \, H_{\beta} \, \nu, $$
where $z_{\beta}$ is a normalization constant, the {\em heat kernel measure} on $G$.

Now consider a compact connected Abelian Lie group $G$ with  Lie algebra $\mathfrak g$ of dimension $n$. The exponential map $\exp:\ \mathfrak g \longrightarrow G
$
is a surjective Lie group homomorphism whose kernel 
$
\Lambda := \ker(\exp) \subset \mathfrak g
$
is a lattice of rank $n$, hence $G$ is isomorphic to the torus $\mathfrak g / \Lambda$. Using this, one obtains a canonical isomorphism 
$\widehat{G} \;\cong\; \Lambda^* \subset \mathfrak g^*$, where $\Lambda^*$ is the dual lattice. Alternatively, one can regard the embedding $\widehat{G} \hookrightarrow \mathfrak g^*$ as given by taking the derivative of a character at identity.

We fix an inner product on $\mathfrak g$ and denote by $\| \cdot \|_{ \mathfrak g^*}$ the induced norm on  $\mathfrak g^*$. Further, we write $\Delta$ and $H_{\beta}$ respectively for the (positive) Laplacian and the heat kernel on $G$ with respect to the induced invariant metric on $G$. The following fact is well-known.
\begin{lemma}\label{hetkerexpa}
The heat kernel $H_{\beta}$ has Fourier transform $$\widehat{H_{\beta}}(\xi) = e^{-\,4\pi^2 \beta\,\|\xi\|_{\mathfrak g^*}^{2}},\quad \xi\in \Lambda^*,$$ hence
\[
H_{\beta}(x) = \sum_{\xi\in\widehat{G}} e^{-\,4\pi^2 \beta\,\|\xi\|_{\mathfrak g^*}^{2}}\,\chi_\xi(x),
\]
where $ \chi_\xi$ is the image of $\xi \in \Lambda^*$ under the isomorphism  $\Lambda^*\cong \widehat{G}$, with absolute and uniform convergence for each fixed $\beta>0$.
\end{lemma}
\begin{proof} (Sketch)
Let $X_i, i=1,\ldots, n$, be an orthonormal basis of $\mathfrak g$ with respect to the given inner product. Then the associated invariant vector fields $\widetilde X_i$ on $G$ form a global orthonormal frame and one has
$ \Delta=-\sum_i \widetilde X^2_i$. Using this, one easily checks that 
$$ 
\Delta \chi_\xi(x) = 4\pi^2\,\|\xi\|_{\mathfrak g^\ast}^2 
\quad \xi \in  \Lambda^*,
$$
which by standard arguments implies the desired conclusion.
\end{proof}
\begin{proposition}\label{hetkertoushp}
Let $G_1$ and $G_2$ be compact connected Abelian Lie groups with Lie algebras $\mathfrak g_1$ and $\mathfrak g_2$ and let $F: G_1 \rightarrow G_2$ be a Lie group homomorphism. Fix inner products on $\mathfrak g_1$ and $\mathfrak g_2$ and write $\mu^1_\beta$ and $\mu^2_\beta$ for the associated heat kernel measures.

Suppose that $(dF_e)^*: \mathfrak g^*_2 \rightarrow \mathfrak g^*_1$ (the dual of the differential of $F$ at the identity) is an isometry. Then $F(\mu^1_\beta)=\mu^2_\beta$ for all $\beta >0$.
\end{proposition}
\begin{proof}
We note that under the identification of $\Lambda^*_i$ with $\widehat{G}_i$, the restriction of $(dF_e)^*$ to the dual lattices coincides with the dual map $\widehat{F}$. Hence,
writing $H_{\beta}^i$ for the associated heat kernel on $G_i$, we find by Lemma \ref{hetkerexpa} that $\widehat{H_{\beta}^2}=\widehat{H_{\beta}^1} \circ \widehat{F}$. The claim now follows from Lemma \ref{pushfofour}.
\end{proof}

\subsection{Pushforward measure properties}
In this section we collect several basic properties of pushforward measures on quotients spaces and topological groups that will be needed in the sequel.
\begin{lemma}\label{pushmequot} (1) Let \(X\) be a compact topological space, \(\mu\) a finite Borel measure on \(X\), and let \(\sim\) be an equivalence
	relation with quotient map \(\pi:X\to Y:=X/\!\sim\). Suppose \(f:X\to[0,\infty)\) is continuous
	and constant on equivalence classes. Then there exists a continuous  \(g:Y\to[0,\infty)\) with
	\(f=g\circ\pi\), and
	$\pi(f\mu) \;=\; g\,\pi(\mu)$.
		
	(2) Let \(X_1,X_2\) be compact spaces, \(F:X_1\to X_2\) a homeomorphism, and \(\nu_1,\nu_2\) finite Borel measures with
	\(F(\nu_1)=\nu_2\). If \(f:X_2\to[0,\infty)\) is continuous, then
	$
	F^{-1}(f\,\nu_2)=(f\circ F)\,\nu_1$.
	
\end{lemma}
\begin{proof}
Both statements follow directly from the change of variables formula and functoriality for pushforward measures.
\end{proof}

\begin{lemma}\label{pushmetopgruo} (1) Let \(G_1\) and \(G_2\) be compact topological groups with normalized Haar measures \(\nu_1,\nu_2\).
	If \(F:G_1\to G_2\) is a surjective continuous group homomorphism, then
	$F(\nu_1)=\nu_2$.
	
	(2) In the setting of Part (1), let \(\sim\) be an equivalence relation on \(G_1\) such that \(F\) is constant on equivalence classes,
	and let \(\pi:G_1\to G_1/\!\sim\) be the quotient map. Then there exists a continuous
	\(\bar F:G_1/\!\sim\;\to G_2\) with \(F=\bar F\circ\pi\), and if \(\bar\nu_1:=\pi(\nu_1)\), then $\bar F(\bar\nu_1)=\nu_2$.
\end{lemma}
\begin{proof}
To establish Part (1), one checks that the probability measure $F(\nu_1)$ is $G_2$-left invariant and hence coincides with the unique normalized Haar measure \(\nu_2\). Part (2) then follows from the functoriality of pushforward measures and Part (1).
\end{proof}

\subsection{Group actions and projective limits}
In what follows, $G$ is an arbitrary group and $\mathbf{Set}$ stands for the category of sets. Let  $(I, \leq)$ be a (possibly uncountable) directed set equipped with a left $G$-action by order preserving automorphisms.
\begin{definition}\label{euqiprojde}
 A \emph{$G$-equivariant projective system in $\mathbf{Set}$} consists of:
\begin{itemize}
	\item a projective system $(\{X_i\}_{i\in I},\{P_{ij}:X_j\to X_i\}_{i\le j})$ in $\mathbf{Set}$, and
	\item for each $g\in G$, a family of maps
	\[
	\phi_{g,i}: X_i \rightarrow X_{g\cdot i},\quad i\in I,
	\]
\end{itemize}
	satisfying the following properties:
	
	(1) (Naturality) For all $i\le j$ in $I$,
		\begin{equation*}
			P_{g\cdot i,\,g\cdot j}\circ \phi_{g,j}
			\;=\;
			\phi_{g,i}\circ P_{ij}
		\end{equation*}
		
	(2) (Unit) $\phi_{e,i}=\Id_{X_i}$ for all $i\in I$.
		
	(3) (Multiplicativity) For all $g,h\in G$ and $i\in I$,
		\begin{equation*}
			\phi_{gh,i} \;=\; \phi_{g,\,h\cdot i}\circ \phi_{h,i}.
		\end{equation*}
\end{definition}
In other words, a $G$-equivariant projective system is a functor from the action groupoid associated to the action of $G$ on $I$ to $\mathbf{Set}$.

We denote the limit of a $G$-equivariant projective system $(\{X_i\},\{P_{ij}\})$ in $\mathbf{Set}$ by $X_\infty$ and for every $g \in G$ define a map $\Phi_g: X_\infty \to X_\infty$ via
 \begin{equation}\label{indproactde}
	\, (\Phi_g(x))_{\,g\cdot i}\;=\;\phi_{g,i}\big(x_i\big), \,\quad x=\{x_i\} \in X_\infty.
\end{equation}
\begin{lemma}\label{lemproact}
The maps $\Phi_g$ are well-defined and make $X_\infty$ into a left $G$-set.
\end{lemma} 
\begin{proof}
The fact that $\Phi_g$ is well-defined follows from property (1) in Definition \ref{euqiprojde}. Using properties (2) and (3), one verifies that  the maps $\Phi_g$ form a $G$-action.
\end{proof}
We recall that arbitrary inductive and projective limits exist in 
$\mathbf{Hilb}_1$, the category of Hilbert spaces and linear contractions, see \cite[Section 2]{Gr}.
\begin{definition}\label{euqiprohilb}
	A \emph{$G$-equivariant projective system in $\mathbf{Hilb}_1$} consists of:
	\begin{itemize}
		\item a projective system $(\{\mathcal{H}_{i}\}_{i\in I},\{P_{ij}\}_{i\le j})$ in $\mathbf{Hilb}_1$, and
		\item for each $g\in G$, a family of linear contractions
		\[
		\phi_{g,i}: \mathcal{H}_i \rightarrow \mathcal{H}_{g\cdot i},\quad i\in I,
		\]
	\end{itemize}
	satisfying the properties (1)--(3) in Definition \ref{euqiprojde}.
\end{definition}

\begin{proposition}\label{proponisomgievi}
The projective limit $\mathcal{H}_\infty$ of a $G$-equivariant projective system $(\{\mathcal{H}_{i}\}_{i\in I},\{P_{ij}\}_{i\le j})$ in $\mathbf{Hilb}_1$ is endowed with a natural $G$-action by linear contractions. If the maps $\phi_{g,i}$ are isometries, this $G$-action is unitary.
\end{proposition} 
\begin{proof}
One defines the induced action as in (\ref{indproactde}) and utilizes Lemma \ref{lemproact}; the verification that these maps are contractions, respectively isometries is straightforward.
\end{proof}

\section{Villain action in Abelian polyhedral gauge theory}\label{vactinpoly}

\subsection{Basic notions}
Let $G$ be a compact connected Abelian Lie group with Lie algebra $\mathfrak g$ and $K$ a finite connected polyhedral cell complex of dimension $d>1$. We fix orientations of the cells of $K$ and write $C^k(K;G)$ (respectively $C^k(K;\mathfrak g)$) for the groups (vector spaces) of oriented $k$-cochains on $K$ with values in $G$ (respectively $\mathfrak g$). We note that $C^k(K;G)$ is also a compact connected Abelian Lie group (isomorphic to a product of finitely many copies of $G$) with Lie algebra $C^k(K;\mathfrak g)$ for every $k \geq 0$.

We shall make use of the usual cellular coboundary operators
$$ d_k^{\hspace{1pt}G}: C^k(K;G)\to  C^{k+1}(K;G),$$
$$ d_k^{\hspace{1pt}\mathfrak g}: C^k(K;\mathfrak g)\to  C^{k+1}(K;\mathfrak g),$$
for $k=0,1$ (see e.g. \cite[Chapter IX]{Mas} for a definition using incidence numbers), which are Lie group homomorphisms.

We now suppose we are given an inner product on $C^2(K;\mathfrak g)$ and write $H^K_{\beta}$ for the associated heat kernel on $C^2(K;G)$. For every $\beta >0$ we define a Borel probability measure
\begin{equation}\label{genvildefmea}
\mu^K_{\beta}=z^K_\beta (H^K_{\beta} \circ d_1^{\hspace{1pt}G}) \cdot \nu^K_1,
\end{equation}
on $C^1(K;G)$, where $z^K_\beta$ is a normalization constant and $\nu^K_1$ is the normalized Haar measure on $C^1(K;G)$.

One readily checks, using the multiplicativity property of heat kernels, that in the special case when $K$ is a cubical lattice and the inner product on $C^2(K;\mathfrak g)$ is induced by the Euclidean product on $\mathfrak g$,  the measure $\mu^K_{\beta}$ coincides with the measure corresponding to the so called Villain action in lattice gauge theory with a coupling parameter $\beta$, as defined, for example in \cite[p. 6]{Sei}. Motivated by this, we shall refer  $\mu^K_{\beta}$ as the measure associated to the polyhedral gauge theory with Villain action on $K$.

We mention that when $G$ is not Abelian, one can still define a non-Abelian coboundary operator $d_1^{\hspace{1pt}G}$ and therefore the analogue of the measure $\mu^K_{\beta}$. This point of view towards non-Abelian lattice gauge theory will be explored in a sequel of this article.

We now observe that $\mathtt{Im}\hspace{1pt}d_1^{\hspace{1pt}G}$, the image of $d_1^{\hspace{1pt}G}$ is a closed subgroup of $C^2(K;G)$ and hence a compact connected Abelian Lie group. Moreover, the Lie algebra of $\mathtt{Im}\hspace{1pt}d_1^{\hspace{1pt}G}$ can naturally be identified with $\mathtt{Im}\hspace{1pt}d_1^{\hspace{1pt}\mathfrak g}$. Since $\mu^K_{\beta}$ depends only the values of $H^K_{\beta}$ on $\mathtt{Im}\hspace{1pt}d_1^{\hspace{1pt}G}$, it is natural to modify (\ref{genvildefmea}) as follows. We fix an inner product on $\mathtt{Im}\hspace{1pt}d_1^{\hspace{1pt}\mathfrak g}$, write $H^I_{\beta}$ for the associated heat kernel on $\mathtt{Im}\hspace{1pt}d_1^{\hspace{1pt}G}$, and set
\begin{equation}\label{genvildefmod}
	\mu^I_{\beta}=z^I_\beta (H^I_{\beta} \circ d_1^{\hspace{1pt}G}) \cdot \nu^I_1,
\end{equation}
where $z^I_\beta$ is a normalization constant and $\nu^K_1$ is the normalized Haar measure on $\mathtt{Im}\hspace{1pt}d_1^{\hspace{1pt}G}$. When $d=2$ the operator $d_1^{\hspace{1pt}G}$ is surjective, hence $\mu^I_{\beta}=\mu^K_{\beta}$. When $d>2$ the restriction of $H^K_{\beta}$ to $\mathtt{Im}\hspace{1pt}d_1^{\hspace{1pt}G}$ does not coincide in general with $H^I_{\beta}$ which implies $\mu^I_{\beta} \neq \mu^K_{\beta}$.

\subsection{Gauge equivalence and pushforward measures}\label{gaueqipussec}
We observe that the usual gauge action of $C^0(K;G)$ on $C^1(K;G)$ can be written as
$$ c_0 \cdot c_1 = (d_0^{\hspace{1pt}G}c_0)c_1, \quad c_0 \in C^0(K;G),\,
c_1 \in C^1(K;G).$$
It follows that one has a canonical identification 
\begin{equation}\label{orbspaceide}
 C^1(K;G)/C^0(K;G) \cong C^1(K;G)/\mathtt{Im}\hspace{1pt}d_0^{\hspace{1pt}G},
\end{equation} 
hence the orbit space  $C^1(K;G)/C^0(K;G)$ is naturally a compact connected Abelian Lie group. Further, since $d_0^{\hspace{1pt}G} d_1^{\hspace{1pt}G}=1$, the operator $d_1^{\hspace{1pt}G}$ induces a map
\begin{equation}\label{mapindbyd}
\bar{d}_1^{\hspace{1pt}G}: C^1(K;G)/C^0(K;G) \to C^2(K;G).
\end{equation} 
\begin{lemma}\label{injofdone} Assume that the first homology group (with integer coefficients) of $K$ vanishes. Then  $\bar{d}_1^{\hspace{1pt}G}$ is injective and maps $C^1(K;G)/C^0(K;G)$ homeomorphically onto $\mathtt{Im}\hspace{1pt}d_1^{\hspace{1pt}G}$.
\end{lemma}
\begin{proof}
By (\ref{orbspaceide}), $\mathtt{Ker}\hspace{1pt}\bar{d}_1^{\hspace{1pt}G}$ 
may be identified with 
$\mathtt{Ker}\hspace{1pt}d_1^{\hspace{1pt}G}
/\mathtt{Im}\hspace{1pt}d_0^{\hspace{1pt}G}$, i.e. with the first cohomology group of $K$ with coefficients in $G$, hence the claim follows by the universal coefficient theorem.
\end{proof}
 We denote  by $\bar{\mu}^I_{\beta}$  the pushforward of $\mu^I_{\beta}$ to $C^1(K;G)/C^0(K;G)$ via the quotient map and by $\mu^H_{\beta}$ the heat kernel measure on $\mathtt{Im}\hspace{1pt}d_1^{\hspace{1pt}G}$ associated to $H^I_{\beta}$.
\begin{proposition}\label{pusfintrept}
Assume that the first homology group of $K$ vanishes. Then one has 
$ \bar{\mu}^I_{\beta}= (\bar{d}_1^{\hspace{1pt}G})^{-1}\mu^H_{\beta}$.
\end{proposition}
\begin{proof}
By Lemma \ref{pushmequot}(1) one has 
$\bar{\mu}^I_{\beta}=z^I_{\beta}( H^I_{\beta}\circ \bar{d}_1^{\hspace{1pt}G})\nu^Q_1$, where $\nu^Q_1$ is the normalized Haar  measure on $C^1(K;G)/C^0(K;G)$. Further, we infer from Lemma \ref{pushmetopgruo} that $\bar{d}_1^{\hspace{1pt}G}(\nu^Q_1)$ coincides with the normalized Haar measure on $\mathtt{Im}\hspace{1pt}d_1^{\hspace{1pt}G}$. The claim now follows from Lemma \ref{pushmequot}(2).
\end{proof}
\section{Projective systems of measures}\label{projmlimmeasec}
\subsection{The general construction}
Let $\{K_i\}_{i=0}^{\infty}$ be a sequence of finite connected polyhedral cell complexes whose first homology group vanishes. We continue to use the notation of Section \ref{vactinpoly} and write
$$ d_{i,k}^{\hspace{1pt}G}: C^k(K_i;G)\to  C^{k+1}(K_i;G),$$
$$ d_{i,k}^{\hspace{1pt}\mathfrak g}: C^k(K_i;\mathfrak g)\to  C^{k+1}(K_i;\mathfrak g),$$
for the respective cellular coboundary operators.
Let $\langle \cdot, \cdot \rangle_i$ be an inner product on $\mathtt{Im}\hspace{1pt}d_{i,1}^{\hspace{1pt}\mathfrak g}$ and let
$$P^k_i: C^k(K_i;G) \to C^k(K_{i-1};G), \quad k=0,1,2,$$
be surjective Lie group homomorphisms forming a cochain map, i.e. satisfying
\begin{equation}\label{cochamapco}
	P^{k+1}_{i}d^G_{i,k} =d^G_{i-1,k} P^k_{i},  \quad  i>0, \, k=0,1.
\end{equation}
We observe that the linear maps between Lie algebras induced by $P^k_i$ also form a cochain map; in particular we obtain surjections
\begin{equation}\label{p2bulledef}
 (P^2_i)_{\bullet}: \mathtt{Im}\hspace{1pt}d_{i,1}^{\hspace{1pt}\mathfrak g} \to \mathtt{Im}\hspace{1pt}d_{i-1,1}^{\hspace{1pt}\mathfrak g}.
\end{equation}
We shall now construct, exactly as in \cite[Section 3.2]{Z1}, new ``renormalized'' inner products $\langle \cdot, \cdot \rangle_{i}^r$ on $\mathtt{Im}\hspace{1pt}d_{i,1}^{\hspace{1pt}\mathfrak g}$ out of the inner products $\langle \cdot, \cdot \rangle_{i}$  turning the maps $(P^2_i)_{\bullet}$ into co-isometries. We set 
$ \langle \cdot, \cdot \rangle_{0}^r=\langle \cdot, \cdot \rangle_{0}$
and, using the decomposition
\begin{equation}\label{basiorthde}
	\mathtt{Im}\hspace{1pt}d_{i,1}^{\hspace{1pt}\mathfrak g}=\mathtt{Im}\hspace{1pt}(P^2_{i})_{\bullet}^* \oplus (\mathtt{Im}\hspace{1pt}(P^2_{i})_{\bullet}^*)^{\perp},
\end{equation} 
for every $i>0$ inductively define $\langle \cdot, \cdot \rangle_{i}^r$   to be the pullback of $\langle \cdot, \cdot \rangle_{i-1}^r$ via $(P^2_{i})_{\bullet}$ on $\mathtt{Im}\hspace{1pt}(P^2_{i})_{\bullet}^*$ and to be equal to $\langle \cdot, \cdot \rangle_{i}$ on $(\mathtt{Im}\hspace{1pt}(P^2_{i})_{\bullet}^*)^{\perp}$.
Finally, we declare the two summands in (\ref{basiorthde}) to be orthogonal with respect to  $\langle \cdot, \cdot \rangle_{i}^r$. (Above, the adjoints and orthogonal complements are taken with respect to the inner products $\langle \cdot, \cdot \rangle_{i}$ restricted to $\mathtt{Im}\hspace{1pt}d_{i,1}^{\hspace{1pt}\mathfrak g}$.)

Next we introduce projective systems of cochains via setting
$$P^k_{ij}:=P^k_{i} P^k_{i-1} \cdots  P^k_{j+1}, \quad i>j, k=0,1,2,$$
and note that 
\begin{equation}\label{cochamapcoij}
	P^{k+1}_{ij}d^{\hspace{1pt}G}_{i,k} =d^{\hspace{1pt}G}_{j,k} P^k_{ij}, \quad i>j,\,k=0,1.
\end{equation}
Following (\ref{genvildefmod}), we define probability measures
\begin{equation}\label{genvildefiver}
	\mu^r_{\beta,i}=z_{\beta,i} (H^r_{\beta,i} \circ d_{i,1}^{\hspace{1pt}G}) \cdot \nu^i_1,
\end{equation}
on $C^1(K_i;G)$, where $z_{\beta,i}$ is a normalization constant, $H^r_{\beta,i}$ stands for the heat kernel on $\mathtt{Im}\hspace{1pt}d_{i,1}^{\hspace{1pt}G}$ associated to the inner product  $\langle \cdot, \cdot \rangle_{i}^r$, and $\nu^i_1$ is the normalized Haar measure on $\mathtt{Im}\hspace{1pt}d_{i,1}^{\hspace{1pt}G}$.
We denote by $\bar{\mu}^r_{\beta,i}$ the measures on $C^1(K_i;G)/C^0(K_i;G)$ induced by $\mu^r_{\beta,i}$ and observe that by 
(\ref{cochamapcoij}) and (\ref{orbspaceide}), $P^1_{ij}$ induces a group homomorphism
$$ \bar{P}^1_{ij}: C^1(K_i;G)/C^0(K_i;G)  \to C^1(K_j;G)/C^0(K_j;G), \quad i>j.$$
We are now ready to formulate the main result of this section.
\begin{theorem}\label{mainconstrthe}
	The triple $(C^1(K_i;G)/C^0(K_i;G),\bar{P}^1_{ij},\bar{\mu}^r_{\beta,i})$ is a projective system of measures.
\end{theorem}
\begin{proof}
We write $\mu^{H,r}_{\beta,i}$ for the heat kernel measure on $\mathtt{Im}\hspace{1pt}d_{i,1}^{\hspace{1pt}G}$ associated to $H^r_{\beta,i}$ and notice that by construction the maps $(P^2_i)_{\bullet}^{*,r}$, where $(*,r)$ stands for the adjoint taken with respect to $\langle \cdot, \cdot \rangle_{i}^r$,  are isometries with respect to $\langle \cdot, \cdot \rangle_{i}^r$. It follows from Proposition \ref{hetkertoushp} that the triple $(\mathtt{Im}\hspace{1pt}d_{i,1}^{\hspace{1pt}G},P^2_{ij},\mu^{H,r}_{\beta,i})$ is a projective system of measures.

We now observe that (\ref{cochamapcoij}) implies
$$ P^{2}_{ij}\bar{d}^{\hspace{1pt}G}_{i,1} =\bar{d}^{\hspace{1pt}G}_{j,1} \hat{P}^1_{ij}, \quad i>j,$$
and hence
\begin{equation}\label{dinverpro}
(\bar{d}^{\hspace{1pt}G}_{j,1})^{-1}P^{2}_{ij}= \bar{P}^1_{ij} (\bar{d}^{\hspace{1pt}G}_{i,1})^{-1}, \quad i>j.
\end{equation}
On the other hand, by Proposition \ref{pusfintrept} one has $\bar{\mu}^r_{\beta,i}= (\bar{d}^{\hspace{1pt}G}_{i,1})^{-1}\mu^{H,r}_{\beta,i}$, which together with (\ref{dinverpro}) implies the desired result. 
\end{proof}

We write $\mathbf{TopG}_c$ for the category of compact topological groups and $(C^1/C^0)_\infty$ for the projective limit in $\mathbf{TopG}_c$  of the projective system \newline $(C^1(K_i;G)/C^0(K_i;G),\bar{P}^1_{ij})$. Recall that (see e.g. \cite[Part I, Chapter I, \S 10]{Sz}), the projective system of measures $(C^1(K_i;G)/C^0(K_i;G),\bar{P}^1_{ij},\bar{\mu}^r_{\beta,i})$ has a limit measure $\mu^r_{\beta,\infty}$ defined on $(C^1/C^0)_\infty$.

\begin{remark}
The projective system of measures $\mu^{H,r}_{\beta,i}$ introduced in the proof of Theorem \ref{mainconstrthe} has a limit measure $\mu^{H,r}_{\beta,\infty}$ defined on $\mathtt{Im}\hspace{1pt}d_{\infty,1}^{\hspace{1pt}G}$, the limit in $\mathbf{TopG}_c$ of the projective system $(\mathtt{Im}\hspace{1pt}d_{i,1}^{\hspace{1pt}G},P^2_{ij})$. 
By properties of projective limits the maps $(\bar{d}^{\hspace{1pt}G}_{i,1})^{-1}$  induce a continuous map $(\bar{d}^{\hspace{1pt}G}_{\infty,1})^{-1}: \mathtt{Im}\hspace{1pt}d_{\infty,1}^{\hspace{1pt}G} \to (C^1/C^0)_\infty$ and one has $\mu^r_{\beta,\infty}=(\bar{d}^{\hspace{1pt}G}_{\infty,1})^{-1}(\mu^{H,r}_{\beta,\infty})$. This alternative description of $\mu^r_{\beta,\infty}$ suggests a different strategy for constructing measures on $(C^1/C^0)_\infty$  representing Abelian gauge fields which will be pursued in Section \ref{limmeassect}.
\end{remark}

\subsection{Examples}
In this section we set $G=U(1)$ and present several applications of Theorem \ref{mainconstrthe} to infinite volume limits and continuum limits.

\begin{example}(Infinite cubical lattice limit)\label{exinfculat}
We take the complexes $K_i$ to be finite contractible cubical sublattices of $\mathbb{Z}^d \subset \mathbb{R}^d$ such that $K_i \subset K_{i+1}$ and $\cup_i K_i=\mathbb{Z}^d$, $P_i^k$ to be the natural restrictions of cochains, and $\langle \cdot, \cdot \rangle_i$ to be the restrictions to  $\mathtt{Im}\hspace{1pt}d_{i,1}^{\hspace{1pt}\mathfrak g}$ of the Euclidean inner products on $C^2(K_i;\mathfrak g)$. The infinite lattice limit measure $\mu^r_{\beta,\infty}$ obtained by Theorem \ref{mainconstrthe} is  different than the usual infinite lattice limit considered in \cite{FS} if $d>2$. In fact, $\mu^r_{\beta,\infty}$ is not even invariant under lattice translations when $d>2$. An alternative construction of an infinite lattice measure which is invariant under lattice translations will be given in Section \ref{altinflatclim}. When $d=2$, the coboundary operators $d^{\hspace{1pt}G}_{i,1}$  are surjective, the maps $P_i^k$ are co-isometries with respect to $\langle \cdot, \cdot \rangle_i$ and $\langle \cdot, \cdot \rangle^r_i=\langle \cdot, \cdot \rangle_i$.
\end{example}

\begin{example}\label{contlimcubexa}(Continuum limit of cubical lattices)
Consider a $d$-dimensional cube $K_0  \subset \mathbb{R}^d$ and take $K_i$ to be the cubical lattices given by iterated subdivisions of $K_0$ such that $K_i$ consists of $2^{id}$ cubes. We take $P_i^0$ to be the natural restriction map and $P_i^k$ for $k=1,2$ to be the standard subdivision maps given by assigning to a coarse $k$-cell the product of the values assigned to its sub-cells belonging to the finer lattice, so that (\ref{cochamapco}) is satisfied. A natural choice for $\langle \cdot, \cdot \rangle_i$ in this setting is the Euclidean inner product on $C^2(K_i;\mathfrak g)$ multiplied by a suitable power of the spacing/mesh of the lattice $K_i$. When $d=2$ one has $\langle \cdot, \cdot \rangle^r_i=\langle \cdot, \cdot \rangle_i$ and the projective system of measures given by Theorem \ref{mainconstrthe} coincides with the one obtained in \cite{VM}.
\end{example}

\begin{example}(Continuum limit of triangulations)
Let $M$ be a compact Riemannian manifold (possibly with boundary) and take $K_i$ a sequence of smooth triangulations of $M$ whose mesh tends to 0 and such that $K_{i+1}$ is a subdivision of $K_{i}$ for each $i$. We take $P_i^k$ to be any subdivision maps obeying (\ref{cochamapco}) (the existence of such maps is well-known, see e.g. \cite[Chapter 2, \S 17]{Mu}). A natural choice for $\langle \cdot, \cdot \rangle_i$ in this setting is the so-called Whitney inner product introduced (in the case of real-valued cochains) in \cite{Do}.
\end{example}

\section{Limit measures via projective limits of Hilbert spaces}\label{limmeassect}

\subsection{Limit heat kernel measures}\label{limkermesse}
We begin by introducing a category whose objects generalize the notion of a compact connected Abelian Lie group equipped with invariant Riemannian metric.
\begin{definition}
Let $\mathbf{LG}$ be the category whose
\begin{itemize}
	\item \emph{objects} are triples $(G,\mathcal H,J)$ with
	$G$ a compact Abelian group, $H$ a real Hilbert space, and $J:\widehat{G}\to \mathcal H^{\ast}$ a homomorphism of Abelian groups;
	\item \emph{morphisms} $(\Psi,\Phi):(G_1,\mathcal H_1,J_1)\to (G_2, \mathcal H_2,J_2)$ are pairs, where
	$\Psi:G_1\to G_2$ is a continuous group homomorphism, $\Phi:\mathcal H_1\to \mathcal H_2$ is a linear contraction, and one has
	\begin{equation}\label{morcompat}
		\Phi^{\ast}\circ J_2\ =\ J_1\circ \Psi^{\ast}.
	\end{equation}
	Composition of morphisms is defined componentwise.
\end{itemize}
\end{definition}
Given an object $(G,H,J)$ in $\mathbf{LG}$, we set 
$$ S_\beta(G,H,J)(\chi):= \exp\big(-4\pi^{2}\beta \|J(\chi)\|_{\mathcal H^{\ast}}^{2}\big), \quad \chi \in \widehat{G},$$
for every $\beta>0$.  By Schoenberg's theorem (\cite[Chapter 3, Theorem 2.2]{BCR}) the function $S_\beta(G,\mathcal H,J)$ is positive definite, hence by  Bochner's theorem for compact Abelian groups (\cite[Theorem 4.19]{Fo}) there exists a unique Borel probability measure $\mu_{\beta}(G,\mathcal H,J)$ on $G$ whose Fourier transform is $S_\beta(G,\mathcal H,J)$. We shall refer to $\mu_\beta(G,\mathcal H,J)$ as the {\em heat kernel measure associated to} $(G,\mathcal H,J)$.

\begin{example}\label{comgrexlg}
We take $G$ to be a compact connected Abelian Lie group, $\mathcal H$ to be the Lie algebra of $G$ equipped with an inner product and $J$ to be the natural inclusion of $\widehat{G}$ into $\mathcal H^*$. Then $\mu_{\beta}(G,\mathcal H,J)$ coincides with the heat kernel measure on $G$ defined in Section \ref{hkermeassec} (cf. Lemma \ref{hetkerexpa}). 
\end{example}
In what follows, $\mathbf{AbG}$ stands for the category of Abelian groups.
\begin{proposition}
The category $\mathbf{LG}$ has projective limits.
\end{proposition}
\begin{proof}
We consider a projective system 
\begin{equation}\label{projsylgi}
((G_i,\mathcal{H}_i,J_i),(\Psi_{ij},\Phi_{ij}))
\end{equation}
in $\mathbf{LG}$ indexed by an arbitrary directed set $I$,
write $G_\infty$ for the projective limit of $(G_i,\Psi_{ij})$ in $\mathbf{TopG}_c$, and $\mathcal{H}_\infty$ for the projective limit of 
$(\mathcal{H}_i,\Phi_{ij})$ in $\mathbf{Hilb}_1$. Further, we denote by $\widehat{G}^{\infty}$ the inductive limit of $(\widehat{G}_i,\Psi^*_{ij})$ in  $\mathbf{AbG}$ and by $(\mathcal{H}^*)^\infty$ the inductive limit of $(\mathcal{H}^*_i,\Phi^*_{ij})$ in $\mathbf{Hilb}_1$.

By properties of inductive limits in  $\mathbf{AbG}$ the maps $J_i$ induce an Abelian group homomorphism $\bar{J}_{\infty}$ from $\widehat{G}^{\infty}$ to $(\mathcal{H}^*)^\infty$. Since Pontryagin duality is a contravariant equivalence of categories (cf. \cite[Theorem 5]{Roe}), $\widehat{G}^{\infty}$ is canonically isomorphic to  $\widehat{G_\infty}$. Using this and the canonical isometric embedding of $(\mathcal{H}^*)^\infty$ into $(\mathcal{H}_\infty)^*$ we obtain from $\bar{J}_{\infty}$ a group homomorphism $J_{\infty}:\widehat{G_\infty}   \to (\mathcal{H}_\infty)^*$. It is straightforward to check that the triple $(G_\infty,J_{\infty},\mathcal{H}_\infty)$ is a projective limit of (\ref{projsylgi}).
\end{proof}

\subsection{The infinite lattice limit}\label{altinflatclim}
Let $G$ be a compact connected Abelian Lie group and $L$ be a fixed polyhedral decomposition of $\mathbb{R}^d$. We denote by $\mathcal{F}_{L}$ the set of all finite cell subcomplexes of $L$ which are homeomorphic to a connected contractible $d$-dimensional manifold with boundary and endow $\mathcal{F}_{L}$ with the partial order given by inclusion. Given $K \in \mathcal{F}_{L}$, we continue to use the notation from Sections \ref{vactinpoly} and \ref{projmlimmeasec} for the groups and Lie algebras of cochains on $K$ and the associated coboundary operators, replacing the index $i$ with $K$.

Recall that for all $K, K' \in \mathcal{F}_{L}$ with $K\subset K'$ one has restriction maps
$$ P^2_{K,K'}: \mathtt{Im}\hspace{1pt}d_{K',1}^{\hspace{1pt}G} \to \mathtt{Im}\hspace{1pt}d_{K,1}^{\hspace{1pt}G}.$$
We fix an inner product $\langle \cdot , \cdot\rangle_0$ on $C^{2}_0(L;\mathfrak g)$, the space of finitely supported $\mathfrak g$-valued 2-cochains on $L$, such that the Lie algebra restriction maps
$$ (P^2_{K,K'})_{\bullet}: \mathtt{Im}\hspace{1pt}d_{K',1}^{\hspace{1pt}\mathfrak g} \to \mathtt{Im}\hspace{1pt}d_{K,1}^{\hspace{1pt}\mathfrak g} $$
are contractions with respect to the restricted inner products. Thus we obtain a projective system
\begin{equation}\label{prpjistlgex}
	((\mathtt{Im}\hspace{1pt}d_{K,1}^{\hspace{1pt}G},\mathtt{Im}\hspace{1pt}d_{K,1}^{\hspace{1pt}\mathfrak g},J_K),( P^2_{K,K'},(P^2_{K,K'})_{\bullet}))
\end{equation}
in $\mathbf{LG}$ indexed by the directed set $\mathcal{F}_{L}$, where $J_K$ is given by Example \ref{comgrexlg}. We denote the limit of (\ref{prpjistlgex}) in $\mathbf{LG}$ by $(\mathtt{Im}\hspace{1pt}d_{L,1}^{\hspace{1pt}G},
\mathtt{Im}\hspace{1pt}d_{L,1}^{\hspace{1pt}\mathfrak g},J_L)$ and the heat kernel measure associated to this object in $\mathbf{LG}$ by $\mu^{H}_{\beta,L}$.

We write $(C^1/C^0)_L$ for the limit in $\mathbf{TopG}_c$  of the projective system $(C^1(K;G)/C^0(K;G),\bar{P}^1_{K,K'})$ indexed by $\mathcal{F}_{L}$ and note that the maps 
 $$(\bar{d}^{\hspace{1pt}G}_{K,1})^{-1}: \mathtt{Im}\hspace{1pt}d_{K,1}^{\hspace{1pt}G} \to C^1(K;G)/C^0(K;G),$$
defined as in Section \ref{gaueqipussec}, induce a continuous map $(\bar{d}^{\hspace{1pt}G}_{L,1})^{-1}: \mathtt{Im}\hspace{1pt}d_{L,1}^{\hspace{1pt}G} \to (C^1/C^0)_L$. Finally, we set 
\begin{equation}\label{meuinterldefi}
\mu_{\beta, L}=: (\bar{d}^{\hspace{1pt}G}_{L,1})^{-1}\mu^{H}_{\beta,L}.
\end{equation}

Let us now provide a more explicit description of the above construction in the special case when $L=\mathbb{Z}^d \subset \mathbb{R}^d$ and $\langle \cdot , \cdot\rangle_0$ is the (Euclidean) $\ell^2$-inner product. We denote by $C^{k}_{(2)}(\mathbb{Z}^d;\mathfrak g)$ the Hilbert space of $\ell^2$ $k$-cochains on $\mathbb{Z}^d$ and consider the coboundary 
$$d_1^{\hspace{1pt}\mathbb{Z}^d}: C^{1}_{(2)}(\mathbb{Z}^d;\mathfrak g) \to C^{2}_{(2)}(\mathbb{Z}^d;\mathfrak g) $$
as a bounded operator between Hilbert spaces. Then the projective limit of $(C^2(K;\mathfrak g),(P^2_{K,K'})_{\bullet})$ in $\mathbf{Hilb}_1$  can  be naturally identified with $C^{2}_{(2)}(\mathbb{Z}^d;\mathfrak g)$. Moreover, it is easy to check that the projective limit of $(\mathtt{Im}\hspace{1pt}d_{K,1}^{\hspace{1pt}\mathfrak g},
(P^2_{K,K'})_{\bullet})$ in $\mathbf{Hilb}_1$ is isomorphic to 
$\overline{\mathtt{Im} \hspace{1pt} d_1^{\hspace{1pt}\mathbb{Z}^d}}$. Further, identifying the Hilbert spaces  $\mathtt{Im}\hspace{1pt}d_{K,1}^{\hspace{1pt}\mathfrak g}$  with their duals, we see that the inductive limit of $((\mathtt{Im}\hspace{1pt}d_{K,1}^{\hspace{1pt}\mathfrak g})^*,
(P^2_{K,K'})^*_{\bullet})$ is also naturally isomorphic to $\overline{\mathtt{Im} \hspace{1pt} d_1^{\hspace{1pt}\mathbb{Z}^d}}$ and the canonical maps $E_K: (\mathtt{Im}\hspace{1pt}d_{K,1}^{\hspace{1pt}\mathfrak g})^* \to
\overline{\mathtt{Im} \hspace{1pt} d_1^{\hspace{1pt}\mathbb{Z}^d}}$ are given by extending a cochain by 0 and then orthogonally projecting onto 
$\overline{\mathtt{Im} \hspace{1pt} d_1^{\hspace{1pt}\mathbb{Z}^d}}$.
It follows that $E_K$ and hence $J_L$ are injective.

We note that $\mathbb{Z}^d$ acts by lattice translations on $C^{k}_{(2)}(\mathbb{Z}^d;\mathfrak g)$, on the groups $C^{k}(\mathbb{Z}^d;G)$ of $G$-valued cochains on $\mathbb{Z}^d$, as well as on $(C^1/C^0)_{\mathbb{Z}^d}$ since the latter can be identified with 
$C^{1}(\mathbb{Z}^d;G)/C^{0}(\mathbb{Z}^d;G)$.

\begin{proposition}
The measure $\mu_{\beta, \mathbb{Z}^d}$ defined in (\ref{meuinterldefi}) is invariant under lattice translations.
\end{proposition}
\begin{proof}
We observe that the obvious action of  $\mathbb{Z}^d$ on $\mathcal{F}_{\mathbb{Z}^d}$ by translations turns $(\mathtt{Im}\hspace{1pt}d_{K,1}^{\hspace{1pt}\mathfrak g},
(P^2_{K,K'})_{\bullet})$ and $(C^1(K;G)/C^0(K;G),\bar{P}^1_{K,K'})$
into $\mathbb{Z}^d$-equivariant projective systems (cf. Definition \ref{euqiprojde}) and the induced (via Lemma \ref{lemproact} and Proposition \ref{proponisomgievi}) $\mathbb{Z}^d$-actions on the projective limits coincides with the natural  ones mentioned above. It follows that the Fourier transform of the measure $\mu^{H}_{\beta,\mathbb{Z}^d}$, and hence $\mu^{H}_{\beta,\mathbb{Z}^d}$  is $\mathbb{Z}^d$-invariant.  Since the coboundary operator on $\mathbb{Z}^d$ is equivariant with respect to lattice translations, the same holds for $\mu_{\beta, \mathbb{Z}^d}$.
\end{proof}

\subsection{Absence of mass gap}
In this section we assume that $G=U(1)$ and study the connected 2-point correlation function associated to the measure $\mu_{\beta, \mathbb{Z}^d}$.
We observe that given a 2-cell $p$ the evaluation map of a cochain in $\mathtt{Im}\hspace{1pt}d_{K,1}^{\hspace{1pt}G}$ at $p$ is 
is a character of $\mathtt{Im}\hspace{1pt}d_{K,1}^{\hspace{1pt}G}$ for all $K \in \mathcal{F}_{\mathbb{Z}^d}$ such that $ p \in K$. Passing to the projective limit, these evaluations produce a character of $\mathtt{Im}\hspace{1pt}d_{\mathbb{Z}^d,1}^{\hspace{1pt}G}$ which we shall denote by $\chi_p$.

We write 
$W_p:C^1(\mathbb{Z}^d;G) \to G$
for the Wilson loop corresponding to the boundary of $p$, so that 
$W_p(c)= \chi_p(d_{\mathbb{Z}^d,1}^{\hspace{1pt}G}(c))$ for all $c \in C^1(\mathbb{Z}^d;G)$, and  define a connected 2-point correlation function
\begin{equation}
	O_{\beta}(p,q)=\int \overline{W}_p \overline{W}_q d\mu_{\beta, \mathbb{Z}^d} - \left(\int\overline{W}_p d\mu_{\beta, \mathbb{Z}^d}\right)\left(\int\overline{W}_q d\mu_{\beta, \mathbb{Z}^d}\right),
\end{equation}
where $p$ and $q$ are 2-cells in $\mathbb{Z}^d$, the integrals are taken over $C^1(\mathbb{Z}^d;G) /C^0(\mathbb{Z}^d;G)$, and $\overline{W}_p$ stands for the map on $C^1(\mathbb{Z}^d;G) /C^0(\mathbb{Z}^d;G)$ induced by $W_p$.
\begin{proposition}\label{propoofsq} One has
	$$ O_{\beta}(p,q)=e^{-4\pi^2\beta(\|\Pi \delta_p \|^2_{\ell^2}+\|\Pi \delta_q \|^2_{\ell^2})}\left(e^{-8\pi^2\beta\langle  \delta_p,\Pi \delta_q \rangle_{\ell^2}} -1\right),$$
		where $\delta_p \in C^{2}_{(2)}(\mathbb{Z}^d;\mathbb R)$ stands for the 2-cochain which is 1 at $p$ and 0 elsewhere, and $\Pi$ is the orthogonal projection onto $\overline{\mathtt{Im} \hspace{1pt} d_1^{\hspace{1pt}\mathbb{Z}^d}}$.
\end{proposition}
\begin{proof}
A change of variables yields
\begin{equation}\label{chofvaropq}
	O_{\beta}(p,q)=\int \chi_p \chi_q d\mu^H_{\beta,\mathbb{Z}^d} - \left(\int \chi_p d\mu^H_{\beta,\mathbb{Z}^d}\right)\left(\int  \chi_q d\mu^H_{\beta,\mathbb{Z}^d}\right),
\end{equation}
where the integrals are taken over $\mathtt{Im}\hspace{1pt}d_{\mathbb{Z}^d,1}^{\hspace{1pt}G}$. By construction one has 
\begin{equation}\label{jconneltwo}
\| J^{\mathbb{Z}^d}(\chi_p) \|_{*}= \| \Pi \delta_p \| _{\ell^2},\end{equation}
where $\| \cdot \|_{*}$ stands for the norm on the dual space of $\overline{\mathtt{Im} \hspace{1pt} d_1^{\hspace{1pt}\mathbb{Z}^d}}$. Since
$$ \int \chi_p d\mu^H_{\beta,\mathbb{Z}^d}=\exp\big(-4\pi^{2}\beta \|J^{\mathbb{Z}^d}(\chi_p)\|_{*}^{2}\big) $$
and
$$ \int \chi_p \chi_q d\mu^H_{\beta,\mathbb{Z}^d}=
\exp\big(-4\pi^{2}\beta 
\|J^{\mathbb{Z}^d}(\chi_p)+J^{\mathbb{Z}^d}(\chi_q)\|_{*}^{2}\big),$$
the desired formula easily follows from (\ref{chofvaropq}) and (\ref{jconneltwo}).
\end{proof}

We note that when $d=2$ the measure $\mu_{\beta, \mathbb{Z}^d}$ coincides with the usual thermodynamic limit of the $U(1)$-lattice gauge theory and in this case, as expected, Proposition \ref{propoofsq} implies that $ O_{\beta}(p,q)=0$ unless $p=q$.

From now on we assume that $d>2$, fix a 2-cell $p$ in $\mathbb{Z}^d$, a lattice direction $e$, and for every $n \in \mathbb{N}$, write $p+ne$ for the cell $p$ translated in the direction $e$ by $n$ units.

\begin{lemma}\label{decaylemm} The sequence $|\langle  \delta_p,\Pi \delta_{p+ne} \rangle_{\ell^2}|$ decays to 0 not faster than $n^{-d}$ as $n \to \infty$.
\end{lemma}
\begin{proof} In the course of this proof we shall freely use the notation and results from Appendix \ref{appedsec}. We write $d_k$ for short for the coboundary operators $d_k^{\hspace{1pt}\mathbb{Z}^d}$ and observe that $\overline{\mathtt{Im} \hspace{1pt} d_k}=\mathtt{Ker} \hspace{1pt} d_{k+1}$. Moreover, the Laplacians $\Delta_k=d_k^{*} d_k + d_{k-1} d^*_{k-1}$ are invertible and one has 
\begin{equation}\label{projkerform}
	\Pi = \text{Id} - d_2^{*} \Delta^{-1}_2 d_2.
\end{equation}
One verifies (\ref{projkerform}) using the identities 
$$  \Delta_{k+1} d_k= d_k\Delta_k, \quad d^*_k\Delta_{k+1}=\Delta_kd^*_k .$$
Further, Lemma \ref{symtrainop} implies
\begin{equation}
\langle  \delta_p,\Pi \delta_{p+ne} \rangle_{\ell^2}=
\frac{1}{(2\pi)^d}\int_{\mathbb{T}^d} e^{\,in (e \cdot \xi)\cdot \xi}\ M^{\Pi-\text{Id}}_{ij,ij}(\xi)\ d\xi,
\end{equation}
where $M^{\Pi-\text{Id}}$ is the symbol of $\Pi-\text{Id}$ and the indices $i$ and $j$ are such that the basis vectors $e_i$ and $e_j$ span the (oriented) 2-cell $p$. Using Lemma \ref{symofsoop} one finds that $M^{\Pi'}_{ij,ij}(\xi)$ is equal, up to a phase factor, to 
$$F_0(\xi):=\frac{\displaystyle\sum_{r\notin\{i,j\}}\sin^2\!\big(\tfrac{\xi_r}{2}\big)}{\displaystyle \sum_{r=1}^d \sin^2\!\big(\tfrac{\xi_r}{2}\big)}  .$$
We set
$$F_1(\xi)= \sum_{r\notin\{i,j\}}\xi_r^2/\|\xi\|^2, \quad \xi \in \mathbb{R}^d,$$
and note that both $F_0$ and $F_1$ have (essential) discontinuity at 0. 

The Fourier transform $\hat{F}_1(x)$ of $F_1(\xi)$ (regarded as a tempered homogeneous distribution) can be explicitly computed and decays like $\|x\|^{-d}$ as $\|x\| \to \infty$. Let $\chi$ be a smooth compactly supported real valued function on $\mathbb{R}^d$ which is 1 near the origin. It is not hard to check that $\chi F_0-F_1$ is smooth, hence the Fourier transform of $\chi F_0$ obeys the same power law decay from which the desired conclusion easily follows.
\end{proof}
The following theorem implies that the lattice field associated to the measure $\mu_{\beta, \mathbb{Z}^d}$ is massless for all $\beta>0$.
\begin{theorem}\label{masslessth}
The sequence $|O_{\beta}(p,p+ne)|$ tends to 0 not faster than $n^{-d}$ as $n \to \infty$ for all $\beta>0$.
\end{theorem}
\begin{proof}
The statement follows from Proposition \ref{propoofsq} and Lemma \ref{decaylemm} in view of the elementary inequality
$$|e^{-t}-1|\ \ge\ \frac{|t|}{1+|t|},\quad t\in\mathbb{R}.$$
\end{proof}
	
\subsection{A continuum limit construction}
In this section we continue to assume that $G=U(1)$ and describe a construction of an Euclidean invariant continuum limit measure representing an Abelian gauge theory on a suitable projective limit space.

We denote by $\mathcal{Q}^d$ the set of all pairs $(K,T)$, where $K$ is a $d$-cube in $\R^d$ and $T$ is a triangulation of $K$. We write $(K,T)\preceq (K',T')$ if (1) $K \subseteq K'$, (2) $K$ is a union of $d$-simplices in $T'$, and (3) the restriction of $T'$ to $K$ is a subdivision of $T$. It is straightforward to check that $(\mathcal{Q}^d,\preceq)$ is a partially ordered set. Moreover, one has:
\begin{lemma}
The set $(\mathcal{Q}^d,\preceq)$ is directed.
\end{lemma}
\begin{proof}
Let $(K,T), (K',T') \in \mathcal{Q}^d$, $K''$ be a cube such that $K \cup K ' \subset K''$, and $T^0$ be any triangulation of $K''$. Then by \cite[Addendum 2.12]{RS} we can find a subdivision $T''$ of $T^0$ such that 
$(K,T)\preceq (K'',T'')$ and $(K',T')\preceq (K'',T'')$.
\end{proof}
If $(K,T), (K',T') \in \mathcal{Q}^d$ with $(K,T)\preceq (K',T')$, one has natural restriction/subdivision maps
$$ P^k_{T,T'}: C^k(T';G) \to C^k(T;G), $$
$$ (P^k_{T,T'})_{\bullet}: C^k(T';\mathfrak{g}) \to C^k(T;\mathfrak{g})$$
for $k=1,2$, given by first restricting a $k$-cochain to $T$ and then  multiplying/adding its values over the fine \(k\)-simplices subdividing each coarse \(k\)-simplex of \(T\), with signs determined by orientation. 
One has
$$P^2_{T,T'} d_{T',1}^{\hspace{1pt}G}=   d_{T,1}^{\hspace{1pt}G}P^1_{T,T'},$$
$$(P^2_{T,T'})_{\bullet}d_{T',1}^{\hspace{1pt}\mathfrak g}=d_{T,1}^{\hspace{1pt}\mathfrak g}(P^1_{T,T'})_{\bullet}, $$
where $d_{T,1}^{\hspace{1pt}G}$ and $d_{T,1}^{\hspace{1pt}\mathfrak g}$ are the corresponding coboundary operators acting on 1-cochains, hence the maps $ P^k_{K,K'}$ and $(P^k_{K,K'})_{\bullet}$
preserve the images of $d_{T,1}^{\hspace{1pt}G}$ and $d_{T,1}^{\hspace{1pt}\mathfrak g}$, respectively. 

Given $(K,T)\in \mathcal{Q}^d$, we introduce an area-normalized inner product on $C^2(T;\mathfrak{g})$ by setting
\begin{equation}\label{areanorminn}
\langle c_1, c_2 \rangle_T=\sum_{p \in T^{(2)}} \frac{1}{A(p)}c_1(p)c_2(p),\quad c_1,c_2 \in C^2(T;\mathfrak{g}),
\end{equation}
where $T^{(2)}$ is the set of all (unoriented) 2-simplices in $T$, $A(p)$ denotes the area of the 2-simplex $p$, we have identified $\mathfrak{g}$ with $\R$ and products of  values of the oriented cochains $c_1$ and $c_2$ are computed using any choice of orientations of the 2-simplices. (Note that the expression (\ref{areanorminn}) is independent of this choice of orientations.) As above, we write $J_T$ for the canonical inclusion of  $\widehat{\mathtt{Im}\hspace{1pt}d_{T,1}^{\hspace{1pt}G}}$ into $(\mathtt{Im}\hspace{1pt}d_{T,1}^{\hspace{1pt}\mathfrak g})^*$.
\begin{lemma}
The collection
\begin{equation}\label{prpjistll}
	\left\{(\mathtt{Im}\hspace{1pt}d_{T,1}^{\hspace{1pt}G},\mathtt{Im}\hspace{1pt}d_{T,1}^{\hspace{1pt}\mathfrak g},J_T),( P^2_{T,T'},(P^2_{T,T'})_{\bullet})\right\}_{(K,T) \in \mathcal{Q}^d},
\end{equation} 
where $\mathtt{Im}\hspace{1pt}d_{T,1}^{\hspace{1pt}\mathfrak g}$ is endowed with the inner product (\ref{areanorminn}),    is a projective system in $\mathbf{LG}$.
\end{lemma}
\begin{proof}
The identity $(P^2_{T,T'})^*_{\bullet}\circ J_{T'} =J_{T}\circ (P^2_{T,T'})^*$  follows from the naturality of the inclusion of the dual of a compact connected Abelian Lie group into the dual of its Lie algebra with respect to Lie group homomorphisms. It remains to show that $(P^2_{T,T'})_{\bullet}$ is  contraction. To this end, let $(K,T), (K',T') \in \mathcal{Q}^d$ with $(K,T)\preceq (K',T')$ and let $p$ be a 2-simplex in $T$. Then, choosing orientations of $p$ and the 2-simplices in $T'$ contained in $p$, we can write
$$
((P^2_{T,T'})_{\bullet}c)(p)=\sum_{q\subset p}\varepsilon(q,p)c(q), \quad c \in C^2(T';\mathfrak{g}),
$$
where the sum is over the 2-simplices in $T'$ contained in $p$ and the signs $\varepsilon(q,p)=\pm 1$ are determined by the choice of orientations. The Cauchy-Schwartz inequality implies
$$
\left|\sum_{q\subset p}\varepsilon(q,p)c(p)\right|^2
\leq
\left(\sum_{q\subset p}A(q)\right)
\left(\sum_{q\subset p}\frac{|c(q)|^2}{A(q)}\right),
$$
hence, noticing that $\sum_{q\subset p}A(q)=A(p)$, one has 
\begin{equation}\label{fracserin}
\frac{|(P^2_{T,T'})_{\bullet}c)(p)|^2}{A(p)}
\leq
\sum_{q\subset p}\frac{|c(q)|^2}{A(q)}.
\end{equation}
Finally, summing (\ref{fracserin}) over all 2-simplices in in $T$, we obtain  
$$\|((P^2_{T,T'})_{\bullet}c)\|^2_{T} \leq \|c\|^2_{T'}.$$
\end{proof}

 We denote the limit of (\ref{prpjistll}) in $\mathbf{LG}$ by $(\mathtt{Im}\hspace{1pt}d_{\infty,1}^{\hspace{1pt}G},
\mathtt{Im}\hspace{1pt}d_{\infty,1}^{\hspace{1pt}\mathfrak g},J_{\infty})$ and the heat kernel measure associated to this object in $\mathbf{LG}$ by $\mu^{H}_{\beta,\infty}$.

Next we observe that the maps $ P^1_{T,T'}$ descend to maps
$$  \bar{P}^1_{T,T'}: C^1(T';G)/C^0(T';G)  \to  C^1(T;G)/C^0(T;G) $$
and write $(C^1/C^0)_{\infty}$ for the limit of the 
projective system 
$$\left\{(C^1(T;G)/C^0(T;G),\bar{P}^1_{T,T'})\right\}_{T \in \mathcal{Q}^d}$$ in $\mathbf{TopG}_c$. Finally, we note that the maps
$$(\bar{d}^{\hspace{1pt}G}_{T,1})^{-1}: \mathtt{Im}\hspace{1pt}d_{T,1}^{\hspace{1pt}G}
 \to C^1(T;G)/C^0(T;G) $$
defined in Section \ref{altinflatclim} induce a continuous map
$$(\bar{d}^{\hspace{1pt}G}_{\infty,1})^{-1}: \mathtt{Im}\hspace{1pt}d_{\infty,1}^{\hspace{1pt}G} \to (C^1/C^0)_{\infty}$$
and set
$$\mu_{\beta, \infty}=: (\bar{d}^{\hspace{1pt}G}_{\infty,1})^{-1}(\mu^{H}_{\beta,\infty}).$$
The natural action of the $d$-dimensional Euclidean group $E(d)$ on $\mathcal{Q}^d$ preserves the partial order and turns
$(\mathtt{Im}\hspace{1pt}d_{T,1}^{\hspace{1pt}G},
P^2_{T,T'})$,
 $(\mathtt{Im}\hspace{1pt}d_{T,1}^{\hspace{1pt}\mathfrak g},
(P^2_{T,T'})_{\bullet})$ and $(C^1(T;G)/C^0(T;G),\bar{P}^1_{T,T'})$
into $E(d)$-equivariant projective systems (cf. Definition \ref{euqiprojde}). Thus we obtain (via Lemma \ref{lemproact} and Proposition \ref{proponisomgievi}) $E(d)$-actions on the projective limits
$\mathtt{Im}\hspace{1pt}d_{\infty,1}^{\hspace{1pt}G}$,
$\mathtt{Im}\hspace{1pt}d_{\infty,1}^{\hspace{1pt}\mathfrak g}$ and $(C^1/C^0)_{\infty}$.
\begin{proposition}
The measure $\mu_{\beta, \infty}$ is $E(d)$-invariant.
\end{proposition}
\begin{proof}
By construction, the Fourier transform of the measure $\mu^{H}_{\beta,\infty}$, and hence $\mu^{H}_{\beta,\infty}$  is $E(d)$-invariant. Since the coboundary operators commute with the maps defining the $E(d)$-equivariant projective system structure (cf. Definition \ref{euqiprojde}) on cochains, the same holds for $\mu_{\beta, \infty}$.
\end{proof}

\appendix
\section{Discrete exterior calculus on $\mathbb{Z}^d$ and Fourier multipliers}\label{appedsec}
In this Appendix we utilize the notation introduced in Section \ref{limmeassect} regarding the cubical lattice $\mathbb{Z}^d\subset\mathbb{R}^d$. Denote by $\{e_i\}$ the standard oriented orthonormal basis in $\mathbb{R}^d$ and by $\{e^{k}_\alpha\}$, where $\alpha=(\alpha_1,\ldots,\alpha_k)$ with $\alpha_1<\dots<\alpha_k$, the induced orthonormal basis in the $k$-th exterior power $\Lambda^k(\mathbb{R}^d)$ equipped with the induced Euclidean product. We write $\ell^2(\mathbb{Z}^d;\Lambda^k(\mathbb{R}^d))$ for the Hilbert space of $\Lambda^k(\mathbb{R}^d)$-valued square summable 0-cochains on $\mathbb{Z}^d$.

Given an oriented $k$-cell $p$ in $\mathbb{Z}^d$, we write $v_p$ for the vertex in $p$ closest to the origin and $\alpha_p$ for the multi-index corresponding to the subset of $\{e_i\}$ consisting of the vectors spanning $p$ translated to the origin. We denote by $U(p)$ the element of $\ell^2(\mathbb{Z}^d;\Lambda^k(\mathbb{R}^d))$ defined by assigning
the vector $\{e^{k}_{\alpha_p}\}$ to $v_p$ and 0 to all other vertices in $\mathbb{Z}^d$. The verification of the following lemma is straightforward.

\begin{lemma}\label{ideexterandk}
The assignment $p \mapsto U(p)$ extends to an isometric isomorphism between
$C^k_{(2)} (\mathbb{Z}^d;\mathbb{R})$
and $\ell^2(\mathbb{Z}^d;\Lambda^k(\mathbb{R}^d))$ which interwines the natural actions of $\mathbb{Z}^d$ by lattice translations.
\end{lemma}

For the following, recall that the (discrete) Fourier transform for vector valued functions
$$ \mathcal{F}: \ell^2(\mathbb{Z}^d;\mathbb{C}^r)\to L^2(\mathbb{T}^d;\mathbb{C}^r)$$ is an isometric isomorphism. It is a classical result (see e.g. \cite[I.8.5]{Kat}) that any bounded translation invariant operator $A: \ell^2(\mathbb{Z}^d;\mathbb{C}^{r_1}) \to \ell^2(\mathbb{Z}^d;\mathbb{C}^{r_2})$ is a (matrix-valued) Fourier multiplier, i.e. there exists a matrix-valued function $M^A\in L^\infty(\mathbb{T}^d;\mathrm{Hom}(\mathbb{C}^{r_1},\mathbb{C}^{r_2}))$, called the {\em symbol} of $A$, such that $A=\mathcal{F}^{-1} M^A\,\mathcal{F}$ (here $M^A$ stands also for the multiplication by $M^A$).
\begin{lemma}\label{symtrainop}
Let $A$ be a bounded $\mathbb{Z}^d$-invariant operator on $C^k_{(2)} (\mathbb{Z}^d;\mathbb{C})$ and let $\tilde{A}$ be the induced operator on  $\ell^2(\mathbb{Z}^d;\Lambda^k(\mathbb{C}^d))$ given by Lemma \ref{ideexterandk}. Let $p$ and $q$ be two $k$-cells in $\mathbb{Z}^d$, and $\delta_p$ and $\delta_q$ be the corresponding Dirac cochains in $C^k_{(2)} (\mathbb{Z}^d;\mathbb{C})$. Then
$$\langle \delta_{p}, A\,\delta_{q}\rangle\ = \frac{1}{(2\pi)^d}\int_{\mathbb{T}^d} e^{\,i (v_q-v_p)\cdot \xi}\ M^{\tilde A}_{\alpha_p,\alpha_q}(\xi)\ d\xi,$$ 
where $$M^{\tilde A}_{\alpha_p,\alpha_q}(\xi)=\langle e^{k}_{\alpha_p}, M^{\tilde A}(\xi)e^{k}_{\alpha_q}\rangle_{\Lambda^k(\mathbb{C}^d)}.$$
\end{lemma}
\begin{proof}
One computes as follows:
\begin{align*}
	\langle \delta_{p}, A\,\delta_{q}\rangle
	&=\left\langle U\delta_{p},\, \tilde{A}\,U\delta_{q}\right\rangle_{\ell^2(\mathbb{Z}^d;\Lambda^k(\mathbb{C}^d))} \nonumber\\
	&=\left\langle \mathcal{F}U\delta_{p},\, M^{\tilde{A}}(\xi)\,\mathcal{F}U\delta_{q}\right\rangle_{L^2(\mathbb{T}^d;\Lambda^k(\mathbb{C}^d))} \nonumber\\
	&=\frac{1}{(2\pi)^d}\int_{\mathbb{T}^d} \left\langle e^{-i v_p\cdot \xi} e^{k}_{\alpha_p},\ M^{\tilde{A}}(\xi)\, e^{-i v_q\cdot \xi} e^{k}_{\alpha_q}\right\rangle_{\Lambda^k(\mathbb{C}^d)}\, d\xi \nonumber\\
	&=\frac{1}{(2\pi)^d}\int_{\mathbb{T}^d} e^{\,i (v_q-v_p)\cdot \xi}\ \left\langle e^{k}_{\alpha_p},\ M^{\tilde{A}}(\xi)\, e^{k}_{\alpha_q}\right\rangle_{\Lambda^k(\mathbb{C}^d)}\, d\xi. 
\end{align*}
\end{proof}
Let us denote by $\tilde{d_k}$, $\tilde{d^*_k}$ and $\tilde{\Delta}_k$ the operators on 0-cochains induced via Lemma \ref{ideexterandk} by the coboundary operator $d_k$ on complex-valued $k$-cochains, its adjoint $d^*_k$ and the Laplacian 
$\Delta_k=d_k^{*} d_k + d_{k-1} d^*_{k-1}$, respectively. We set
$$ m(\xi) = \sum_{j=1}^d (e^{i\xi_j} - 1) e_j, \quad \xi=(\xi_1,\ldots, \xi_d ) \in \mathbb{T}^d.$$
\begin{lemma}\label{symofsoop}
The symbols of $\tilde{d_k}$, $\tilde{d^*_k}$ and $\tilde{\Delta}_k$ are given respectively by exterior multiplication by $m(\xi)$, interior multiplication by the conjugate $\overline{m(\xi)}$ and $\left( \sum_{j=1}^d 4\sin^2(\xi_j/2) \right) \cdot \text{Id}$.
\end{lemma}
\begin{proof}
We note that the forward difference operator $\partial_j$ on 0-cochains defined via $\partial_j c(v) = c(v+e_j) - c(v)$ has symbol $e^{i\theta_j} - 1$. One finds that
$$ \tilde{d_k}c = \sum_{j=1}^d e_j \wedge (\partial_j c),
 \quad c \in  \ell^2(\mathbb{Z}^d;\Lambda^k(\mathbb{R}^d)),$$
 where the wedge product is applied vertex-wise. Taking this into account, we obtain
$$ \mathcal{F}(\tilde{d}c)(\xi) = \sum_{j=1}^d e_j \wedge \mathcal{F}(\partial_j c)(\xi) = \sum_{j=1}^d e_j \wedge  (e^{i\xi_j} - 1) \mathcal{F}(c)(\xi),$$
hence the symbol of $\tilde{d_k}$ is exterior multiplication by $m(\xi)$ and the symbol of $\tilde{d^*_k}$ is interior multiplication by  $\overline{m(\xi)}$. Further, we compute
\begin{multline*}
\|m(\xi)\|^2  = \sum_{j=1}^d (e^{-i\xi_j} - 1)(e^{i\xi_j} - 1) 
= \sum_{j=1}^d (2 - 2\cos\xi_j) = 4\sum_{j=1}^d \sin^2(\xi_j/2), 
\end{multline*}
and conclude, using a standard exterior algebra identity, that the symbol of $\tilde{\Delta}_k$ is $\left( 4\sum_{j=1}^d\sin^2(\xi_j/2) \right) \cdot \text{Id}$.
\end{proof}

\end{document}